\documentclass[11pt]{article}
\usepackage{amsmath}   
\usepackage{algorithm}
\usepackage{algpseudocode}
\usepackage{mathrsfs}
\usepackage[utf8]{inputenc}
\usepackage{lmodern}
\usepackage{float}
\usepackage{caption}
\usepackage[normalem]{ulem}
\usepackage{float}
\usepackage[margin=1in]{geometry}
\usepackage{enumitem}
\usepackage{tocloft}
\usepackage{graphicx, titlesec}
\graphicspath{{/figures/}} 
\usepackage{amscd,amsmath,amsthm,amssymb}
\usepackage{verbatim}
\usepackage[all]{xy}
\usepackage{color}
\usepackage{url}
\usepackage[colorlinks=true, linkcolor=purple, citecolor=purple, urlcolor=purple]{hyperref}
\usepackage{booktabs}
\usepackage{tikz-cd}
\usetikzlibrary{patterns}
\usepackage{ytableau}
\usepackage{cleveref}
\usepackage{mathtools}
\usepackage{longtable}
\usepackage{nicematrix}
\usepackage{multirow}
\usepackage{pgf}
\usepackage{todonotes}
\usepackage{xspace,setspace}
\usepackage{bbm,tikz}
\usepackage{algpseudocode}
\usepackage{amsmath,bm}
\usepackage{parallel}

\newcommand{\V}{\bm{V}}

\newcommand{\R}{\mathbb{R}}

\newcommand{\N}{\mathbb{N}}
\newcommand{\Q}{\mathbb{Q}}

\newcommand{\mL}{\mathcal{L}}
\newcommand{\CL}{\mathscr{C}}

\newcommand{\xb}{\mathbf{x}}
\newcommand{\yb}{\mathbf{y}}
\newcommand{\ab}{\mathbf{a}}

\newcommand{\bb}{\mathbf{b}}
\newcommand{\fb}{\mathbf{f}}
\newcommand{\gb}{\bm{g}}
\newcommand{\C}{\mathbb{C}}

\newcommand{\Cn}{\C^n}

\def\1{{\mathbbm 1}}

\newcommand{\gbt}{\widetilde{\gb}}

\newcommand{\compose}{\textsf{Compose}\xspace}

\newcommand{\InRadical}{\textsf{InRadical}\xspace}

\newcommand{\InvariantSet}{\textsf{InvariantSet}\xspace}

\newcommand{\GenerateLoops}{\textsf{GenerateLoops}\xspace}

\theoremstyle{plain}
\newtheorem{theorem}{Theorem}[section]

\newtheorem{proposition}[theorem]{Proposition}

\newtheorem{definition}[theorem]{Definition}
\theoremstyle{definition}
\newtheorem{example}{Example}
\theoremstyle{remark}

\newcommand{\programbox}[2][\linewidth]{
\begin{samepage}\normalfont
\vspace*{.5em}\hspace*{0.3cm}\fbox{
\hspace*{-0.3cm}\begin{minipage}{#1}
\vspace*{-.1em}\begin{algorithmic}
#2
\end{algorithmic}
\vspace*{-.2em}
\end{minipage}
}\\
\end{samepage}
}

\newcommand{\programboxappendix}[2][\linewidth]{
\begin{samepage}\normalfont
\vspace*{.5em}\hspace*{0.0cm}\fbox{
\hspace*{-0.3cm}\begin{minipage}{#1}
\vspace*{-.1em}\begin{algorithmic}
#2
\end{algorithmic}
\vspace*{-.2em}
\end{minipage}
}\\
\end{samepage}
}

\title{Beyond Affine Loops: A Geometric Approach to Program Synthesis 
}
\author{Erdenebayar Bayarmagnai, Fatemeh Mohammadi, and R\'emi Pr\'ebet
}
\date{}

\begin{document}

\maketitle

\begin{abstract}
Ensuring software correctness remains a fundamental challenge in formal program verification. One promising approach relies on finding polynomial invariants for loops. Polynomial invariants are properties of a program loop that hold before and after each iteration. Generating polynomial invariants is a crucial task for loops, but it is an undecidable problem in the general case. Recently, an alternative approach to this problem has emerged, focusing on synthesizing loops from invariants. However, existing methods only synthesize affine loops without guard conditions from polynomial invariants. In this paper, we address a more general problem, allowing loops to have polynomial update maps with a given structure, inequations in the guard condition, and polynomial invariants of arbitrary form.

In this paper, we use algebraic geometry tools to design and implement an algorithm that computes a finite set of polynomial equations whose solutions correspond to all loops satisfying the given polynomial invariants. In other words, we reduce the problem of synthesizing loops to finding solutions of polynomial systems within a specified subset of the complex numbers. The latter is handled in our software using an SMT solver.\end{abstract}

\maketitle
\section{Introduction}

Loop invariants are properties that hold before and after each iteration of a loop. They are key to automating program verification, which ensures that programs produce correct results before execution. Various well-established methods use loop invariants for safety verification, such as the Floyd-Hoare inductive assertion technique \cite{floyd1993assigning} and termination verification via standard ranking functions \cite{manna2012temporal}. When a loop invariant is a polynomial equation or inequality, it is called a polynomial invariant.

In this work, instead of generating polynomial invariants for a given loop, we address the reverse problem, that is synthesizing a loop that satisfies given polynomial invariants.

We consider polynomial loops $\mL(\ab, h, F)$ of the form
\begin{center}
  \programbox[0.7\linewidth]{
\State$\xb:=(x_{1},\ldots, x_n)\gets \ab:=(a_1,\ldots, a_n)$
\While{$h(\xb)\neq 0$}
\State $\begin{pmatrix}
x_1 \\
x_2 \\
\vdots \\
x_n
\end{pmatrix}
\xleftarrow{F}
\begin{pmatrix}
F_1\\
F_2\\
\vdots\\
F_n
\end{pmatrix}
$
\EndWhile
}\label{page:alg}
\end{center}
where $x_i$ are the program variables with initial values $a_i$, $h\in \C[\xb]$, and $F =(F_1,\ldots, F_n)$ is a sequence of polynomials in $\C[\xb]$. 
The inequation $h(\xb) \neq 0$, called the \emph{guard}, is assumed to be a single inequality for simplicity (we can replace $h_1 \neq 0, \ldots, h_k \neq 0$ with their product $h_1 \cdot \ldots \cdot h_k \neq 0$ without loss of generality). Where there is no $h$ we simply write $\mL(\ab, 1, F)$, which corresponds to an infinite loop.

Previous work~\cite{ISSAC2023Laura,hitarth_et_al:LIPIcs.STACS.2024.41,humenberger2022LoopSynthesis} has focused on update maps that are restricted to be linear. In this paper, we go beyond this limitation by allowing update maps to be arbitrary polynomial functions. We present a method for computing a system of equations whose common solutions characterize all coefficient assignments for update maps of loops that satisfy a given set of polynomial invariants. To illustrate our main objective, we now present a motivating example.

\begin{example}\label{ex:introduction}
Consider the following polynomial loop: 

\programboxappendix[0.5\linewidth]{
\State$(x_1, x_2, x_3)\gets(1,1,-1)$
\While{true}
\State $\begin{pmatrix}
x_1 \\
x_2\\
x_3
\end{pmatrix}
\xleftarrow{F}
\begin{pmatrix}
\lambda_1 x_1^3 + \lambda_2 x_2^2 \\
\lambda_3 x_1 + \lambda_4 x_2^2 \\
\lambda_5 x_1
\end{pmatrix}$
\EndWhile
}

\smallskip

\noindent By applying Algorithm~\ref{algo:generateloops}, we construct polynomials $P_1, \ldots, P_4$, 
such that $(\lambda_1, \ldots, \lambda_5)$ is a common root of these polynomials if and only if the loop satisfies the polynomial invariants $g_1 = x_2^2 - x_1$ and $g_2 = x_3^3 + 2x_2^2 - x_1$.  For instance, $(-3, 3, 1, -1, 0)$ is a common root of $\{P_1, \ldots, P_4\}$. Thus, if the update map is $F(x_1, x_2, x_3) = (-3x_1^3 + 3x_2^2, x_1 - x_2^2, 0)$, the loop satisfies the invariants $g_1$ and $g_2$. See Example~\ref{exa:algo2} for further details.
\end{example}

\noindent\textbf{Related works.}
The computation of polynomial invariants for loops has been a prominent area of research over the past two decades \cite{bayarmagnai2024algebraic, de2017synthesizing, hrushovski2018polynomial, karr1976affine, kovacs2008reasoning, kovacs2023algebra, rodriguez2004automatic, rodriguez2007automatic, rodriguez2007generating}. However, computing invariant ideals is undecidable for general loops \cite{hrushovski2023strongest}. As a result, efficient methods have been developed for restricted classes of loops, particularly those where the assertions are linear or can be transformed into linear assertions.

The converse problem of synthesizing loops from given invariants has received less attention, primarily focusing on linear and affine loops. These studies vary in the types of invariants considered: linear invariants in \cite{saurabh2010}, a single quadratic polynomial in \cite{hitarth_et_al:LIPIcs.STACS.2024.41}, and pure difference binomials in \cite{ISSAC2023Laura}. In contrast, \cite{humenberger2022LoopSynthesis} explores general polynomial invariants but lacks the completeness properties of earlier work and does not address guard conditions.

Higher degree loops have been studied in a more general framework in \cite{synthesis2023algebro, Fatemeha}, which also takes polynomial inequalities as input. However, the loops synthesized by this algorithm are restricted to those where the input invariants are inductive, meaning that if the invariant holds after one iteration, it holds for all subsequent iterations.

\smallskip
\noindent\textbf{Our contributions.} 
We consider the problem of generating polynomial loops with guards from arbitrary polynomial invariants. As a first step, we construct a polynomial system whose solutions correspond precisely to loops with a given structure that satisfy the specified invariants. We then discuss strategies for solving this system. 
In practice, in most cases a satisfying solution is found using an SMT solver.

Section~\ref{section:preliminaries} recalls the definition of invariant sets and key results from \cite{bayarmagnaiIssac}. Section~\ref{section:generateloops} introduces a method for identifying multiple polynomial invariants (Proposition~\ref{prop:invarianttest}) and proves that the set $\CL(\ab, h, \fb; \gb)$ of coefficients of polynomial maps of loops satisfying the invariants $\gb$ forms an algebraic variety. We also present Algorithm~\ref{algo:generateloops}, which computes the polynomials defining $\CL(\ab, h, \fb; \gb)$. Section~\ref{sec:polysolve} discusses strategies for solving polynomial systems and their limitations, while Section~\ref{sec:implementation} presents our algorithm's implementation and experimental results.

\section{Preliminaries}\label{section:preliminaries}
In the following, we present some terminology from algebraic geometry. For further details, we refer the reader to \cite{cox2013ideals,kempf_1993, shafarevich1994basic}. 
We denote the field of complex numbers by $\mathbb{C}$. Throughout the paper, $\xb$ denotes indeterminates $x_1, \dots, x_n$, and $\C[\xb]$ the multivariate polynomial ring in these variables. 

\smallskip\noindent{\bf Ideals.} A polynomial ideal $I$ is a subset of $\mathbb{C}[\xb]$ that is closed under addition, $0\in I$ and for any $f\in \C[\xb]$ and $g\in I$, $fg\in I$. Given a subset $S$ of $\C[\xb]$, the ideal generated by $S$ is  
$$\langle S \rangle =\{a_1f_1+\ldots +a_mf_m\mid a_i\in \C[\xb], f_i\in S, m\in \N\}.$$
By Hilbert Basis Theorem~\cite[Theorem 4, Chap.~2]{cox2013ideals}, every ideal is finitely generated. For a subset $X \subset \Cn$, the defining ideal $I(X)$ consists of all polynomials vanishing on $X$. The radical of an ideal $I \subset \C[\xb]$ is defined as 
$$\sqrt{I}=\{f\in \C[\xb]| \,f^m\in I \text{ for some } m\in \N\}.$$

\smallskip\noindent{\bf Varieties.} For $S\subset \C[x]$, the algebraic variety $\V(S)$ is the common zero set of all polynomials in $S$. Moreover, $\V(S)=\V(\langle S\rangle)$ and so every algebraic variety is the vanishing locus of finitely many polynomials. 

\smallskip\noindent{\bf Polynomial maps.} A map $F:\C^n \longrightarrow \C^m$ is a polynomial map if there exist $f_1,\ldots, f_m\in \C[x_1,\ldots, x_m]$ such that $$F(x)=(f_1(x),\ldots, f_m(x))$$ for all $x\in \C^n$. For simplicity, we will identify polynomial maps and their corresponding polynomials.

Here, we introduce the main object of this paper and build upon key results from \cite{bayarmagnaiIssac} by recalling and extending them. 

\begin{definition}\label{def:invariantset}
Let $F : \Cn \longrightarrow \Cn$ be a  map and $X$ be a subset of $\Cn$.
The invariant set of $(F,X)$ is defined as:
\begin{center}
    $S_{(F,X)} = \{x \in X \mid \forall m\in \N, F^{(m)}(x) \in X\},$
\end{center}
where~$F^{(0)}(x)=x$~and~$F^{(m)}(x)= F(F^{(m-1)}(x))$~for~$m>1$.
\end{definition}

The following proposition enables the computation of invariant sets using tools from algebraic geometry, relying on the Hilbert Basis Theorem, which implies that descending chains of algebraic varieties eventually stabilize.

 \begin{proposition}\label{prop:stabilization}
Let $X\subseteq\Cn$ be an algebraic variety and consider a polynomial map $F : \Cn \longrightarrow \Cn$. We define $$X_m=X\cap F^{-1}(X)\cap \ldots \cap F^{-m}(X)$$ for all $m \in \N$. Then, there exists $N\in \N$ such that 
$X_N=X_{N+1}$, and for any such index
$X_N = S_{(F,X)}$.
\end{proposition}

Algorithm~\ref{algo1} computes invariant sets via an iterative outer approximation that converges in finitely many steps, based on Proposition~\ref{prop:stabilization}. It relies on the following procedures:
\begin{itemize}
    \item ``Compose'' takes as input two sequences of polynomials $\gb=(g_1,\ldots, g_m)$ and $F=(F_1,\ldots, F_n)$ in $\Q[\xb]$ and outputs the  polynomials $$\qquad g_1\big(F_1(\xb),\,\ldots,\, F_n(\xb)\big),\ldots, g_m\big(F_1(\xb),\ldots, F_n(\xb)\big).$$ 
    \item ``InRadical'' takes as input a sequence of polynomials $\widetilde{\gb}$ and a finite set $S$ in $\Q[\xb]$ and outputs ``\texttt{True}'' if $\widetilde{\gb}\subset \sqrt{\langle S\rangle}$; ``\texttt{False}'' otherwise.
\end{itemize}
These procedures are classic routines in symbolic computations, and can be performed using various efficient techniques such as Gr\"obner bases \cite{cox2013ideals}. 
See \cite{bayarmagnai2024algebraic} for more details.
\begin{algorithm}[H]
\caption{\InvariantSet}\label{algo1}
\begin{algorithmic}[1]
\Require $\gb$ and~$F = (F_1,\ldots, F_n)$ are sequences in $\mathbb{Q}[\xb]$.
\Ensure Polynomials whose common zero-set is $S_{(F,{\V(\gb)})}$.
\State\label{Step:algo1.1} $S \gets \{\gb\};$
\State\label{Step:algo1.2} $\gbt \gets \compose(\gb,\,F);$
\While{ $\InRadical(\gbt ,\,S)==\texttt{False}$}
\label{Step:algo1.3}\State\label{Step:algo1.4} $S \gets S \cup \{\gbt\};$
\State\label{Step:algo1.5} $\gbt \gets \compose(\gbt,\,F);$
\EndWhile
\State \Return $S$;
\end{algorithmic}
\end{algorithm}
The proof of the termination and correctness of Algorithm~\ref{algo1} follows from Proposition~\ref{prop:stabilization} and 
\cite[Theorem 2.4]{bayarmagnaiIssac}.

\begin{example}
    Let us compute the invariant set of a polynomial map $F(x_1,x_2)=(2x_1-3x_2, x_1+x_2)$ and an algebraic variety $X=\V(x_1^2-x_2^2+x_1x_2)$. On input $F$ and $g=x_1^2-x_2^2+x_1x_2$, Algorithm~\ref{algo1} computes the invariant set of $F$ and $X$  through the following steps.
    \begin{itemize}
        \item At Step~\ref{Step:algo1.1}, $S$ gets $\{g\}$ and at Step~\ref{Step:algo1.2}, $\widetilde{g}$ gets $$\text{\compose}(g,F)=5x_1^2-15x_2x_2+5x_2^2.$$
    \item At Step~\ref{Step:algo1.3}, a Gr\"obner basis of the ideal $\langle g, 1 - t\widetilde{g} \rangle$ is computed to verify that  
\[
\InRadical(\widetilde{g}, S) = \texttt{False}.
\]
    \item At Step~\ref{Step:algo1.4}, $S$ gets $\{g,g\circ F\}$ and at Step~\ref{Step:algo1.5}, $\widetilde{g}$ is set to $$\compose(\widetilde{g},F)=-5x_1^2-35x_1x_2+95x_2^2$$
\item This time, $\InRadical(\widetilde{g}, F) = \texttt{True}$, so the while 
loop terminates.    
\end{itemize}
Therefore, the invariant set $S_{(F,\V(g))}$ is the vanishing locus of polynomials $\{g, g\circ F\}$.
\end{example}

\section{From loops to polynomial systems}\label{section:generateloops}
In this section, we present a method for identifying all loops that satisfy given polynomial invariants by formulating them as solutions of a polynomial system via invariant sets. This reduces the problem to solving a polynomial system, which we explore in the next section.

\begin{definition}\label{def:PI}
A~polynomial~$g$ is an invariant of the loop $\mathcal{L}(\ab, h, F)$~if,~for~any~$m \in \mathbb{Z}_{\geq 0}$, either:  
$$g(F^{(m)}(\ab)) = 0,$$  
or there exists $m\in \mathbb{Z}_{\geq 0}$ such that $g(F^{(m)}(\ab)) =h(F^{(m)}(\ab))=0$ and for every $l<m$:  
\[
g(F^{(l)}(\ab)) = 0 \text{ and }  h(F^{(l)}(\ab)) \neq 0.
\]
\end{definition}
\begin{definition}\label{def:InvId}
    Let $\mL(\ab, h, F)$ be a polynomial loop. The set of all polynomial invariants for $\mL(\ab, h, F)$  is called the invariant ideal of $\mL$ and is denoted by $I_{\mL(\ab, h, F)}$. 
\end{definition}

The invariant ideal is an ideal of the ring $\C[x_1,\ldots, x_n]$ where $x_1,\ldots, x_n$ are  program variables~\cite{rodriguez2004automatic}. Let $f_1,\ldots, f_n$ be polynomials in $\C[x_1,\ldots, x_n]$. We denote by $span\{f_1,\ldots, f_n\}$ the vector space they generate.

\begin{definition}\label{def:Ffb}
    Let $\fb_1 =(f_{1,1},\ldots, f_{1,l_1}),\ldots, \fb_n =(f_{n,1},\ldots, f_{n,l_n})$ and $(F_1,\ldots, F_n)$ be sequences of polynomials in $\C[\xb]$ 
    such that for every i, $$F_i =\displaystyle\sum_{j=1}^{l_i} b_{i,j}f_{i,j}$$ for some $b_{i,j}$'s $ \in \C$. Let $\fb=(\fb_1\ldots, \fb_n)$ and define the  map $$F_{\fb,\bb}:\Cn \longrightarrow \Cn$$ with $F_{\fb,\bb}(x)=(F_1(x),\ldots, F_n(x))$.
    
\end{definition}

The object defined below is the primary focus of this paper. We prove that it is an algebraic variety and compute the polynomial equations that define it.
\begin{definition}
    Using the notations of Definition~\ref{def:Ffb}, let 
        $h\in \C[\xb], \gb = (g_1,\ldots, g_m)$ be a sequence of polynomials in $\C[\xb]$, and  $\ab\in \Cn$.
    Then, the 
    polynomial loop, structured by $\fb$, with invariants including $\gb$,
    is defined by the coefficient set:
    \begin{align*}
       \CL(\ab,h, \fb; \gb) = \{\bb\in \C^{l_1+\ldots+l_n} \mid \gb \subset I_{\mL(\ab, h, F_{\fb,\bb})} \}.
    \end{align*}
\end{definition}
In simple words, $\CL(\ab,h, \fb; \gb)$ is the set of all vectors $\bb \in \C^{l_1+\cdots +l_n}$ such that all polynomials in $\gb$ are polynomial invariants of the following loop:
\begin{center}
  \programbox[0.55\linewidth]{
\State$\xb\gets\ab$
\While{$h(\xb) \neq 0$}
\State $\xb \gets F_{\fb,\bb}(\xb)
$
\EndWhile
}\label{page:alg}
\end{center}
We now give a necessary and sufficient condition for checking whether given polynomials are loop invariants. This extends~\cite[Proposition 2.7]{bayarmagnai2024algebraic}, where the following statement was proven for a single polynomial invariant. 
\begin{proposition}\label{prop:invarianttest}
    Let $h,g_1,\ldots, g_m$ be polynomials in $\C[\xb]$. Let $z$ be a new indeterminate and let $$X=\V(zg_1,\ldots, zg_m)\subset \C^{n+1}.$$ Let $F:\Cn \longrightarrow \Cn$ be a polynomial map and define $$G_{h}(\xb,z)=(F(\xb),zh(\xb)).$$ Then, for $\ab \in \Cn$, $\gb \subset  I_{\mL(\ab, h, F)}$ if and only if $(\ab,1)\in S_{(G_h,X)}$.
\end{proposition}
\begin{proof}
    Let $X_i=\V(zg_i)\subset \C^{n+1}$ be an algebraic variety for $i\in \{1,\ldots, m\}$. For any $i\in \{1,\ldots, m\}$, by~\cite[Proposition 2.7]{bayarmagnai2024algebraic}, $g_i \in I_{\mL(\ab, h, F)}$ if and only if $(\ab, 1)\in S_{(G_h, X_i)}$. Therefore, $g_1,\ldots, g_m\in I_{\mL(\ab, h, F)}$ if and only if $$(\ab,1)\in S_{(G_h, X_1)}\cap \ldots \cap S_{(G_h, X_m)}.$$ 
 It follows directly from Definition~\ref{def:invariantset} that
    $$S_{(G_h, X)}=S_{(G_h, X_1)}\cap \ldots \cap S_{(G_h, X_m)}.$$
Thus, $g_1,\ldots, g_m\in I_{\mL(\ab, h, F)}$ if and only if $(\ab,1)\in S_{(G_h, X)}$.
\end{proof}

In the following proposition, we present a necessary and sufficient condition for loops with a given structure to satisfy specified polynomial invariants.

\begin{proposition}\label{prop:loopgenerator}
    Let $\fb_1 =(f_{1,1},\ldots, f_{1,l_1}),\ldots, \fb_n =(f_{n,1},\ldots, f_{n,l_n})$ be sequences of polynomials in $\C[\xb]$ and define $\fb=(\fb_1,\ldots, \fb_n)$. Let $z, y_{1,1},\ldots, y_{1,l_1}, y_{n,1},\ldots, y_{n,l_n}$ be new indeterminates, let $h\in \C[\xb]$ and define the following polynomial map in $\C[\xb,\yb,z]$:
$$G_{\yb,h}(\xb,\yb,z)=\left(\displaystyle\sum_{i=1}^{l_1}y_{1,i}f_{1,i}(\xb),\ldots, \displaystyle\sum_{i=1}^{l_n}y_{n,i}f_{n,i}(\xb),\yb,zh(\xb)\right).$$
     Let $\gb=(g_1,\ldots, g_m)$ be a sequence of polynomials in $ \C[\xb]$, seen as elements of $\C[\xb,\yb,z]$, let $$X =\V(zg_1,\ldots, zg_m)\subset \C^{n+l_1+\ldots+l_n+1}.$$ Then, for any $\ab \in \Cn$, 
$$\CL(\ab,h,\fb;\gb)=\{\bb\in \C^{l_1+\ldots+l_n}\mid (\ab,\bb,1)\in S_{(G_{\yb,h},X)}\}.$$
\end{proposition}
\begin{proof}
    Write $G_{\yb,h}(\xb,\yb,z)=(G_{\yb}(\xb,\yb),zh(\xb))$ where
        $$G_{\yb}(\xb,\yb)=\left(\displaystyle\sum_{i=1}^{l_1}y_{1,i}f_{1,i}(\xb),\ldots, \displaystyle\sum_{i=1}^{l_n}y_{n,i}f_{n,i}(\xb),\yb\right).$$
    By Proposition~\ref{prop:invarianttest}, $g_1,\ldots, g_m$ are polynomial invariants of $\mL((\ab,\bb), h , G_{\yb})$ if and only if $(\ab,\bb,1)\in S_{(G_{\yb,h},X)}$. 
        Since the value of $\yb$ remains $\bb$ after any iteration, we know that $$F^N_{y}(\xb, \bb)= (F^N(\xb),\bb)$$
    for any $N\in \N$. Therefore, $\gb\subset I_{\mL(\ab, h, F_{\fb,\bb})}$ if and only if $\gb\subset  I_{\mL((\ab,\bb), h, G_{\yb})}$, that is $(\ab,\bb,1)\in S_{(G_{\yb,h},X)}$ by above.
\end{proof}
Since the invariant set is an algebraic variety,by~Proposition~\ref{prop:loopgenerator}, $\CL(\ab, h, \fb; \gb)$ is also an algebraic variety. Specifically, its defining equations are obtained by substituting $\xb = \ab$ and $z = 1$ into the system defining the invariant set. 

We now present our main algorithm for computing the polynomials that define $\CL(\ab, h, \fb; \gb)$.

\begin{algorithm}[H]
\caption{\GenerateLoops}\label{algo:generateloops}
\begin{algorithmic}[1]
\Require$ h, g_1,\ldots, g_m, f_{1,1},\ldots, f_{1, l_1},\ldots,f_{n,1},\ldots, f_{n, l_n}$ in $\Q[\xb]$ and $\ab\in \Q^n$.
\Ensure A set of polynomials $\{P_1,\ldots, P_s\}$ such that $\CL(\ab,h,\fb;g)$ is $\V(P_1,\ldots, P_s)$
\State$G_{\yb,h} \gets \left(\displaystyle\sum_{i=1}^{l_1}y_{1,i}f_{1,i}(\xb),\ldots, \displaystyle\sum_{i=1}^{l_n}y_{n,i}f_{n,i}(\xb),\yb, zh(\xb)\right)$;
\State\label{Step:A1.3}$\{Q_1,\ldots, Q_s\} \gets \InvariantSet((zg_1,\ldots, zg_m), G_{\yb,h});$
\State\label{Step:A1.4}$\{P_1(\yb),\ldots, P_s(\yb)\}\gets \{Q_1(\ab,\yb,1),\ldots, Q_s(\ab,\yb,1)\} $
\State\Return $\{P_1,\ldots, P_s\}$;
\end{algorithmic}
\end{algorithm}
We now prove the correctness of Algorithm~\ref{algo:generateloops}.
\begin{theorem}
    Let $h\in \C[\xb], \gb=(g_1,\ldots, g_m)$ and $\fb_1=(f_{1,1},\ldots, f_{1, l_1}),\ldots,\fb_n=(f_{n,1},\ldots, f_{n, l_n})$ be sequences of polynomials in 
    $\Q[\xb]$. Let $\ab\in \Q^n$ and define $\fb=(\fb_1,\ldots, \fb_n)$. On input $\ab,h, \fb,\gb$, Algorithm~\ref{algo:generateloops} outputs a set of polynomials such that the vanishing locus of the polynomials is $\CL(\ab,h,\fb; \gb)$.
\end{theorem}
\begin{proof}
Let $X=\V(zg_1,\ldots, zg_m)\subset \C^{n+l_1+\ldots+l_n+1}.$ At Step~\ref{Step:A1.3}, by \cite[Theorem 2.4]{bayarmagnaiIssac}, \InvariantSet outputs 
$$\{Q_1(\xb,\yb,z),\ldots, Q_s(\xb,\yb,z)\}$$ 
such that $S_{(G_{\yb,h},X)}$ is the common solution of $\{Q_1,\ldots, Q_s\}$. Therefore, by Proposition~\ref{prop:loopgenerator}, $g_1,\ldots, g_m$ are polynomial invariants of $\mL(\ab, h, F_{\fb,\bb})$ if and only if $$Q_1(\ab,\bb,1)=\ldots = Q_s(\ab,\bb,1)=0.
$$ Consequently, $g_1,\ldots, g_m\in I_{\mL(\ab,h, F_{\fb,\bb})}$ if and only if $$P_1(\bb)=\ldots=P_s(\bb)=0$$ which proves the correctness of Algorithm~\ref{algo:generateloops}.
\end{proof}

\begin{example}~\label{exa:algo2}
Consider the loop $\mathcal{L}((1,1,-1), 1, F )$ together with polynomial invariants $g_1$ and $g_2$ from Example~\ref{ex:introduction}. 
  Let us compute all loops with the polynomial map $F$ that satisfy the polynomial invariants $ g_1$ and $ g_2$
That is, we determine all loops of the given form with the precondition $(x_1 = 1, x_2 = 1, x_3 = -1)$ and the postcondition $(x_2^2 - x_1 = 0, x_3^3 + 2x_2^2 - x_1 = 0)$. 
\noindent Let $F = (F_1,F_2,F_3)$, $\gb=\{g_1,g_2\}$ and define $$\fb=\{\{x_1^3,x_2^2\},\{x_1,x_2^2\},\{x_1\}\}.$$ From the structure of a loop, we know that 
$$F_1\in span\{x_1^3,x_2^2\}, F_2\in span\{x_1,x_2^2\}, F_3\in span\{x_1\}.$$
Thus, the input to Algorithm~\ref{algo:generateloops} is $(\{1\}, \gb, \fb, \{1,1,-1\}).$ Let $X=\V(zg_1, zg_2)\subset \C^9$ and define a polynomial map $$G_{\yb,h}(\xb,\yb,z)=(y_1x_1^3+y_2x_2^2,y_3x_1+y_4x_2^2,y_5x_1,\yb,z)$$

At Step~\ref{Step:A1.3}, on input $G_{\yb,h}$ and $X$, ``InvariantSet'' computes polynomials whose vanishing locus is the invariant set $S_{(G_{y,h},X)}$ and outputs $\{Q_1(\xb,\yb,z),\ldots, Q_6(\xb,\yb,z)\}$. After substituting the initial values of the loop into the polynomials, we obtain only four non-zero polynomials $P_1(\yb),\ldots, P_4(\yb)$ whose vanishing locus is $\CL((1,1,-1),1,\fb;\gb)$:
$$
P_1 = (y_3+y_4)^2-y_1-y_2; \quad  P_2 = y_5^3+2(y_3+y_4)^2-y_1-y_2;
$$
$$
P_3 =2y_3^4y_4^2+8y_3^3y_4^3+12y_3^2y_4^4+8y_3y_4^5+2y_4^6+y_1^3y_5^3+3y_1^2y_2y_5^3+3y_1y_2^2y_5^3+y_2^3y_5^3+4y_1y_3^3y_4+4y_2y_3^3y_4+8y_1y_3^2y_4^2$$
$$+8y_2y_3^2y_4^2+4y_1y_3y_4^3+4y_2y_3y_4^3-y_1^4-3y_1^3y_2-3y_1^2y_2^2-y_1y_2^3+2y_1^2y_3^2+4y_1y_2y_3^2+2y_2^2y_3^2-y_2y_3^2-2y_2y_3y_4-y_2y_4^2;
$$
\begin{flushleft}
$
P_4 = 
y_3^4y_4^2+4y_3^3y_4^3+6y_3^2y_4^4+4y_3y_4^5+y_4^6+2y_1y_3^3y_4+2y_2y_3^3y_4+4y_1y_3^2y_4^2+4y_2y_3^2y_4^2+2y_1y_3y_4^3+2y_2y_3y_4^3-y_1^4-3y_1^3y_2-3y_1^2y_2^2-y_1y_2^3+y_1^2y_3^2+2y_1y_2y_3^2+y_2^2y_3^2-y_2y_3^2-2y_2y_3y_4-y_2y_4^2.
$
\end{flushleft}
\end{example}

\section{Polynomial system solving}\label{sec:polysolve}
In the previous section, we have seen how Algorithm~\ref{algo:generateloops} can compute a system of multivariate polynomials whose solutions correspond exactly to the loops with the given structures and invariants. This approach reduces the problem of loop synthesis to solving polynomial systems.

However, since we aim to find loops with finite, exact representations on computers, we focus on \emph{rational solutions}. 
Unlike solutions in $\C$ (by Hilbert's Nullstellensatz \cite{Hilbert1893}, see e.g., \cite[Chap. 4, \S 1]{cox2013ideals}) or $\R$ (via Tarski-Seidenberg’s theorem \cite{Tar1951, Sei1954}, see e.g., \cite{bpr2006}), determining whether a multivariate polynomial equation has a rational solution is a major open problem in number theory \cite{Shla2011}. For integer solutions, this is Hilbert's Tenth Problem, which is undecidable \cite{Shla2011}. We outline three main strategies to address this, though none is fully satisfactory or complete.

\subsection{Exploit structure}
Although no general algorithm is known to compute (or even decide the existence of) rational solutions to multivariate polynomials using exact methods, many real-world cases can still be successfully addressed with existing techniques. To illustrate this, consider the polynomial system obtained at the end of Example~\ref{exa:algo2}. While this example may appear overly simplistic or too specific, it turns out that all benchmarks presented in Section~\ref{sec:implementation} can be tackled in a similar manner.

\begin{example}  
The variety $\V(P_1, \dotsc, P_4)$ can be decomposed into finitely many irreducible components, which in turn decompose the system into potentially simpler subsystems. This can be done using classical methods from computer algebra \cite[Chap. 4, \S 6]{cox2013ideals}, and here we apply the ``\texttt{minimalPrimes}'' command from Macaulay2~\cite{M2}. We obtain the following five components:

\medskip
\begin{enumerate}
    \item $\V(y_5,y_3+y_4,y_1+y_2)$ 
    \item $\V(y_5+1,y_3+y_4-1,y_1+y_2-1)$
    \item $\V(y_5+1,y_3+y_4+1,y_1+y_2-1)$
    \item $\V(y_3+y_4-1,y_1+y_2-1, y_5^2-y_5+1)$
    \item $\V(y_3+y_4+1,y_1+y_2-1, y_5^2-y_5+1)$
\end{enumerate}
\vspace{3mm}

\noindent Thus, $\mL((1,1,-1),1,F)$ satisfies the polynomial invariants $\{g_1,g_2\}$ if and only if $(\lambda_1,\ldots, \lambda_5)$ lies in one of the above irreducible components. 
The well-structured polynomials defining each component allow us to fully solve the problem.  

For the first three components, which involve only linear polynomials, the problem reduces to a standard linear algebra routine, benefiting from optimized methods (see e.g., \cite{CS2011}). More specifically, if $\mu_1$ and $\mu_2$ are parameters in $\Q$, we obtain the following loop map:
\begin{enumerate}
    \item $F_1(x_1,x_2,x_3)= \big(\mu_1(x_1^3-x_2^2), \mu_2(x_1-x_2^2), 0\big)$
    \item $F_2(x_1,x_2,x_3)= \big(\mu_1x_1^3+(1-\mu_1)x_2^2, \mu_2x_1+(1-\mu_2)x_2^2, -1\big)$
    \item $F_3(x_1,x_2,x_3)= \big(\mu_1x_1^3-(1+\mu_1)x_2^2, \mu_2x_1+(1-\mu_2)x_2^2, -1\big)$
\end{enumerate}

\smallskip\noindent For the last two components, there is no rational solution, because the last equation $y_5^2-y_5+1$ has no rational roots.
\end{example}

Another specific case occurs when the polynomial system computed by Algorithm~\ref{algo:generateloops} (or a subsystem of it) has only finitely many solutions. Geometrically, this means the associated variety (or an irreducible component) consists of finitely many points, i.e., has \emph{dimension zero}. In this case, such systems can be reduced (see e.g., \cite[\S 3]{DL2008} and \cite{Rou1999}) to polynomial systems with rational coefficients of the form:
$$
Q_1(x_1) = 0 \quad \text{and} \quad x_i = Q_i(x_1) \quad \text{for all } i \geq 2.
$$
This reduces to finding all rational solutions of a univariate polynomial, which is either done using arithmetic-based modular algorithms \cite{Lo1983} or efficient factoring algorithms to identify its linear factors in $\Q[x_1]$ \cite[Theorem 15.21]{von2013modern}.

Thus, when the system has finitely many solutions, we can always find all rational ones. Efficient software, such as \textsf{AlgebraicSolving.jl}\footnote{\url{https://algebraic-solving.github.io}}, based on the msolve library \cite{BES2021}, is designed for such tasks. Many of the benchmarks in the next section involve polynomial systems of this type, which we expect when the degree and support of $F$ are close.

\subsection{Numerical methods}
Numerical methods for solving polynomial systems, such as HomotopyContinuation.jl~\cite{HomotopyContinuation}, provide powerful tools for approximating solutions.
These methods use numerical algebraic geometry, particularly homotopy continuation. 
    Unlike symbolic approaches that rely on Gr\"obner bases or resultants, numerical methods can efficiently handle large and complex systems. Using these methods, we can find numerical solutions of the polynomial system generated by Algorithm~\ref{algo:generateloops} and check if close integers or rational numbers
    are valid solutions of the polynomial system. 
\begin{example}\label{ex4}
Consider the polynomial system $\{P_1, \dotsc, P_4\}$ from Example~\ref{exa:algo2}. This system either has no solutions or infinitely many, as there are 5 variables and only 4 equations. A classic approach to handle this is to introduce a random linear form with integer coefficients. 
   Using HomotopyContinuation.jl, we find 15 numerical real solutions to the augmented system, 12 of which correspond to valid integer solutions.
\end{example}

Despite their efficiency, numerical methods have several drawbacks. First, they provide approximate solutions, which may lack exact algebraic structure and require further validation. Second, these methods can struggle with singular solutions. Additionally, homotopy continuation techniques may occasionally miss solutions.

\subsection{Satisfiability Modulo Theories (SMT) Solvers}

SMT solvers determine the satisfiability of logical formulas combining Boolean logic with mathematical theories, such as integer or real arithmetic. By verifying the satisfiability of $$\exists x_1 \ldots \exists x_n (f_1(x_1, \ldots, x_n) = 0) \wedge \ldots \wedge (f_m(x_1, \ldots, x_n) = 0),$$
one can determine the existence of integer or rational solutions to a system of polynomial equations.

SMT solvers leverage automated reasoning, efficient decision procedures, and integration with other logical theories. However, they can struggle with high-degree polynomials and may fail to find general integer solutions due to undecidability. Nevertheless, as discussed in Section~\ref{sec:implementation}, the SMT solver \texttt{Z3} we used successfully finds integer solutions for most systems of polynomials generated by Algorithm~\ref{algo:generateloops}.

\section{Implementation and Experiments
}\label{sec:implementation}
In this section, we present the implementation of Algorithm~\ref{algo:generateloops} and report experiments on benchmarks from related works~\cite{ISSAC2023Laura,hitarth_et_al:LIPIcs.STACS.2024.41,humenberger2022LoopSynthesis}, which are limited to generating linear loops. The implementation in Macaulay2~\cite{M2} is available at
\begin{center}\footnotesize
\href{https://github.com/Erdenebayar2/Synthesizing_Loops.git}{\texttt{https://github.com/Erdenebayar2/Synthesizing\_Loops.git}}
\end{center}
After running Algorithm~\ref{algo:generateloops}, we use the SMT solver \texttt{Z3}~\cite{z3solver} to find a common nonzero integer solution for the polynomials.

\subsection{Implementation details}
The experiments below were performed on a laptop equipped with a 4.8 GHz Intel i7 processor, 16 GB of RAM and 25 MB L3 cache. This prototype implementation is primarily based on the one in \cite{bayarmagnaiIssac} to which we refer for further details.
\subsection{Experimental results}
Let $\fb_1 = (f_{1,1}, \ldots, f_{1,l_1}), \ldots, \fb_n = (f_{n,1}, \ldots, f_{n,l_n})$ and $\gb = (g_1, \ldots, g_m)$ be sequences of polynomials in $\C[\xb]$, and define $\fb = (\fb_1, \ldots, \fb_n)$. Let $h \in \C[\xb]$. 
In Tables~\ref{table1} and~\ref{table2}, we present the outputs and timings for Algorithm~\ref{algo:generateloops}, which computes polynomials whose vanishing locus is $\CL(\ab, h, \fb; g)$, along with timings for \texttt{Z3}~\cite{z3solver} for finding a common non-zero integer solution to the polynomials generated by Algorithm~\ref{algo:generateloops}. The first row shows the benchmarks, while the first column describes the structures of the polynomial maps of loops. 
The benchmarks sources are available at 

\begin{center}\footnotesize
\href{https://github.com/Erdenebayar2/Synthesizing_Loops/tree/master/software/loops}
{\scalebox{.92}{\texttt{https://github.com/Erdenebayar2/Synthesizing\_Loops/software/loops}}}
\end{center}

\vspace{0em}

Table~\ref{table1} presents the execution timings for Algorithm~\ref{algo:generateloops} and \texttt{Z3}~\cite{z3solver}, both using polynomial equations generated by Algorithm~\ref{algo:generateloops}. Timings are in seconds, with a timeout of 300 seconds. ``F'' indicates \texttt{Z3} failed to find an integer solution, while ``NI'' denotes no input was provided to \texttt{Z3} when the algorithm reached the time limit.

\medskip

In the following tables, $n$ denotes the number of program variables, $m$ is the number of polynomial invariants $\gb$ and $d$ is the maximal degree of polynomial invariants $\gb$. Moreover, $D$ denotes the maximal degrees of polynomials in $\fb_1, \ldots, \fb_n$, and $l=l_1+ \cdots+ l_n$.

\medskip

In most cases, when Algorithm~\ref{algo:generateloops} terminates, \texttt{Z3} quickly finds a common non-zero integer solution for the generated polynomials within 0.3 seconds. Thus, finding a non-zero integer solution is not a bottleneck in our approach. The primary computational bottleneck lies in generating the polynomial systems for loop generation. Additionally, we observe that as more polynomial invariants are provided, Algorithm~\ref{algo:generateloops} terminates faster.

\newcommand{\Fexp}{F}
\newcommand{\NIexp}{NI}

\begin{table}[H]
    \centering\hspace*{-0.2cm}
    \scalebox{0.6}{
    \begin{tabular}{|*{4}{c|}*{5}{|c|c|}}
    \hline
        \multicolumn{4}{|c||}{Polynomial map} & \multicolumn{2}{|c||}{D=1, $l=3$} & \multicolumn{2}{|c||}{D=1, $l=4$} &  \multicolumn{2}{|c||}{D=1, $l=5$} & \multicolumn{2}{|c||}{D=2, $l=2$} & \multicolumn{2}{|c|}{D=2, $l=3$} \\ \hline
    Benchmark &$n$ & $m$& $d$ &Alg.~\ref{algo:generateloops} & \texttt{Z3} & Alg.~\ref{algo:generateloops} & \texttt{Z3} & Alg.~\ref{algo:generateloops} & \texttt{Z3} & Alg.~\ref{algo:generateloops} & \texttt{Z3} & Alg.~\ref{algo:generateloops} & \texttt{Z3}\\ \hline      
        Ex2 & 2&1 & 4 & {0.02} & \Fexp & {31.1} & \Fexp & TL & \NIexp & {0.01} & 0.06 &TL&\NIexp \\ \hline
        Ex3 & 3&2 & 3 & {0.01} & 0.06 & {0.04} & 0.06 & 4.4 & 0.2 & {0.01} & 0.06 &0.018 & 0.07 \\ \hline
        Ex3Ineq & 3&2 & 3 & {0.009} & 0.06 & {11.4} & 0.06 & TL & \NIexp & {0.008} & 0.05 &0.03 & 0.05 \\ \hline
        Ex4 & 2&1 & 2 & {0.03} & 0.17 & {0.16} & TL & TL & \NIexp & {0.01} & 0.07 &0.13&0.07\\ \hline
          sum1 & 3&2 & 2 & {0.02} & 0.06 & {0.13} & 0.06 & 1.1 & 0.06 & {0.01} & 0.05 &0.02 & 0.06 \\ \hline
        square & 2&1 & 2 & {0.01} & 0.06 & {0.5} & TL & TL & \NIexp & {0.01} & 0.06 &0.015 & 0.26 \\ \hline
        square\_conj & 3&2 & 2 & {0.019} & 0.06 & {0.14} & 0.09 & 0.39 & 0.06 & {0.007} & 0.05 &0.015 & 0.06 \\ \hline
        fmi1 & 2&1 & 2 & {1.19} & 0.06 & {161.74} & 0.06 & TL & \NIexp & {237.1} & 0.06 &TL & \NIexp \\ \hline
       fmi2 & 3&2 & 2 & {0.01} & 0.06 & {0.015} & 0.06 & 0.017 & 0.07 & {0.009} & 0.07 &0.01 & 0.07 \\ \hline
      fmi3 & 3&2 & 2 & {0.01} & 0.06 & {0.03} & 0.05 & 0.39 & 0.09 & {0.01} & 0.06 &0.2 & 0.11 \\ \hline
       intcbrt & 3&2 & 2 & {0.25} & 0.07 & {16.6} & 0.08 & TL & \NIexp & {0.01} & 0.06 &0.65 &0.17 \\ \hline
cube\_square & 3&1 & 3 & {0.01} & 0.07 & {0.018} &TL& TL & \NIexp & {0.15} & 0.06 &0.96 & 106 \\ \hline
       
    \end{tabular}
    }
    \caption{\label{table1}Timings for Algorithm~\ref{algo:generateloops} in seconds; 
    }
\end{table}

Table~\ref{table2} presents the output of Algorithm~\ref{algo:generateloops}, which is a polynomial system whose solution set exactly corresponds to $\CL(\ab,h,\fb;g)$. For each choice of $D$ and $l$, we report the number $s$ of non-zero polynomials in the output system and indicate whether or not there are finitely many solutions.

\newcommand{\Fin}{$<\!\!\infty$}
\newcommand{\NS}{$\#$sols\xspace}

\begin{table}[H]
    \centering\hspace*{-0.2cm}
    \scalebox{0.65}{
    \begin{tabular}{|*{4}{c|}*{5}{|c|c|}}
    \hline
        \multicolumn{4}{|c||}{Polynomial map} & \multicolumn{2}{|c||}{D=1, $l=3$} & \multicolumn{2}{|c||}{D=1, $l=4$} &  \multicolumn{2}{|c||}{D=1, $l=5$} & \multicolumn{2}{|c||}{D=2, $l=2$} & \multicolumn{2}{|c|}{D=2, $l=3$} \\ \hline
    Benchmark &$n$ & $m$& $d$ &$s$ & \NS & $s$ & \NS & $s$ & \NS & $s$ & \NS & $s$ & \NS\\ \hline      
        Ex2 & 2&1 & 4 & 3 & $\infty$ & 4 & $\infty$ & TL & TL & 2 & $\infty$ &TL&TL \\ \hline
        Ex3 & 3&2 & 3 & 2 & $\infty$ & 4 & $\infty$ & 6 & $\infty$ & 4 & \Fin &4& \Fin\\ \hline
        Ex3Ineq & 3&2 & 3 & 2 & $\infty$ & 4 &$\infty$ & TL & TL & 4 & \Fin &4& \Fin\\ \hline
        Ex4 & 2&1 & 2 & 3 & $\infty$ & 3 &$\infty$ & TL & TL & 3 & \Fin &3&$\infty$\\ \hline
          sum1 & 3&2 & 2& 4 & $\infty$ & 6 & $\infty$ & 6 & $\infty$ & 2 & \Fin &4&\Fin \\ \hline
        square & 2&1 & 2 & 2 & $\infty$ & 4 & $\infty$ & TL & TL & 2 & \Fin &3&$\infty$ \\ \hline
        square\_conj & 3&2 & 2 & 4 & \Fin & 6 & $\infty$ & 6 & $\infty$ & 4 & \Fin &4&\Fin \\ \hline
        fmi1 & 2&1 & 2 & 0 & $\infty$ & 0& $\infty$ & TL & TL & 0 & $\infty$ &TL&TL \\ \hline
       fmi2 & 3&2 & 2 & 2 & $\infty$ & 4 & $\infty$ & 4 & $\infty$ & 2 & \Fin &4& \Fin\\ \hline
      fmi3 & 3&2 & 2 & 4 & \Fin & 4 & $\infty$ & 6 & $\infty$ & 2 &\Fin &4&\Fin \\ \hline
       intcbrt & 3&2 & 2& 4& \Fin & 4 & $\infty$ & TL & TL & 2 & \Fin &4&\Fin\\ \hline
cube\_square & 3&1 & 3 & 2 & $\infty$ & 3 & $\infty$ & TL & TL &3 &\Fin &5&\Fin \\ \hline
       
    \end{tabular}
    }
    \caption{\label{table2} Data on outputs of Algorithm~\ref{algo:generateloops} 
    }
\end{table}

We observe that Macaulay2~\cite{M2} computed the irreducible decompositions of varieties defined by the polynomials generated by Algorithm~\ref{algo:generateloops} in most cases within a few seconds. In many instances, the irreducible components of the varieties are defined by linear equations. In Table~\ref{table2}, even when the dimension of the varieties is 0, the varieties are not empty in all cases, indicating that they consist of finitely many points.

\section{Conclusion}
We presented a method for generating polynomial loops with specified structures that satisfy given polynomial invariants, reducing the problem to solving a system of polynomial equations. Along with an algorithm for constructing this system, we discussed various strategies for solving it. While previous work has primarily focused on linear maps, our approach extends beyond this limitation, marking a significant advancement in the synthesis of non-linear programs.

This work opens up several promising avenues for future research. A natural next step would be to explore the synthesis of more complex loop structures, such as branching and nested loops. Additionally, extending the approach to handle polynomial invariants in the form of inequalities, rather than just equations, could greatly enhance the flexibility and applicability of the method.

\medskip 
\noindent
\textbf{Acknowledgments.}~We are grateful to Teresa Krick for her very helpful comments on the first version of this article. The third author is also affiliated with Inria, CNRS, ENS de Lyon, Université Claude Bernard Lyon 1, LIP (UMR 5668), 69342 Lyon Cedex 07, France.
The first two 
authors are partially supported by the KU Leuven grant iBOF/23/064, the FWO grants G0F5921N and G023721N and the third author is partially supported by 
the ANR grant ANR-20-CE48-0014.

\bibliographystyle{abbrv}
\bibliography{LoopGenerator/references}

\begin{thebibliography}{10}

\bibitem{bpr2006}
S.~Basu, R.~Pollack, and M.-F. Roy.
\newblock {\em Algorithms in Real Algebraic Geometry}.
\newblock Algorithms and Computation in Mathematics. Springer International Publishing, 2nd revised and extended 2016 edition, 2006.

\bibitem{bayarmagnaiIssac}
E.~Bayarmagnai, F.~Mohammadi, and R.~Pr\'{e}bet.
\newblock Algebraic tools for computing polynomial loop invariants.
\newblock In {\em Proceedings of the 2024 International Symposium on Symbolic and Algebraic Computation}, ISSAC '24, page 371–381, New York, NY, USA, 2024. Association for Computing Machinery.

\bibitem{bayarmagnai2024algebraic}
E.~Bayarmagnai, F.~Mohammadi, and R.~Pr{\'e}bet.
\newblock Algebraic tools for computing polynomial loop invariants (extended version).
\newblock {\em arXiv preprint arXiv:2412.14043}, 2024.

\bibitem{BES2021}
J.~Berthomieu, C.~Eder, and M.~{Safey El Din}.
\newblock {msolve: A Library for Solving Polynomial Systems}.
\newblock In {\em {2021 International Symposium on Symbolic and Algebraic Computation}}, pages 51--58, Saint Petersburg, Russia, July 2021. {ACM}.

\bibitem{HomotopyContinuation}
P.~Breiding and S.~Timme.
\newblock Homotopycontinuation.jl: A package for homotopy continuation in julia.
\newblock In J.~H. Davenport, M.~Kauers, G.~Labahn, and J.~Urban, editors, {\em Mathematical Software -- ICMS 2018}, pages 458--465, Cham, 2018. Springer International Publishing.

\bibitem{CS2011}
W.~Cook and D.~E. Steffy.
\newblock Solving very sparse rational systems of equations.
\newblock {\em ACM Trans. Math. Softw.}, 37(4), Feb. 2011.

\bibitem{cox2013ideals}
D.~Cox, J.~Little, and D.~O'Shea.
\newblock {\em Ideals, varieties, and algorithms: an introduction to computational algebraic geometry and commutative algebra}.
\newblock Springer Science \& Business Media, 2013.

\bibitem{z3solver}
L.~De~Moura and N.~Bj\o{}rner.
\newblock Z3: an efficient smt solver.
\newblock In {\em Proceedings of the Theory and Practice of Software, 14th International Conference on Tools and Algorithms for the Construction and Analysis of Systems}, TACAS'08/ETAPS'08, page 337–340, Berlin, Heidelberg, 2008. Springer-Verlag.

\bibitem{de2017synthesizing}
S.~de~Oliveira, S.~Bensalem, and V.~Prevosto.
\newblock Synthesizing invariants by solving solvable loops.
\newblock In {\em International Symposium on Automated Technology for Verification and Analysis}, pages 327--343. Springer, 2017.

\bibitem{DL2008}
C.~Durvye and G.~Lecerf.
\newblock A concise proof of the kronecker polynomial system solver from scratch.
\newblock {\em Expositiones Mathematicae}, 26(2):101--139, 2008.

\bibitem{floyd1993assigning}
R.~Floyd.
\newblock Assigning meanings to programs.
\newblock In {\em Program Verification: Fundamental Issues in Computer Science}, pages 65--81. Springer, 1993.

\bibitem{Fatemeha}
A.~Goharshady, S.~Hitarth, F.~Mohammadi, and H.~Motwani.
\newblock Algebro-geometric algorithms for template-based synthesis of polynomial programs (full version including appendices).
\newblock 2023.

\bibitem{synthesis2023algebro}
A.~K. Goharshady, S.~Hitarth, F.~Mohammadi, and H.~J. Motwani.
\newblock Algebro-geometric algorithms for template-based synthesis of polynomial programs.
\newblock {\em Proceedings of the ACM on Programming Languages}, 7(OOPSLA1):727--756, 2023.

\bibitem{M2}
D.~R. Grayson and M.~E. Stillman.
\newblock Macaulay2, a software system for research in algebraic geometry, 2002.

\bibitem{Hilbert1893}
D.~Hilbert.
\newblock Ueber die vollen invariantensysteme.
\newblock {\em Mathematische Annalen}, 42(3):313--373, 1893.

\bibitem{hitarth_et_al:LIPIcs.STACS.2024.41}
S.~Hitarth, G.~Kenison, L.~Kov\'{a}cs, and A.~Varonka.
\newblock {Linear Loop Synthesis for Quadratic Invariants}.
\newblock In {\em 41st International Symposium on Theoretical Aspects of Computer Science (STACS 2024)}, volume 289, pages 41:1--41:18, 2024.

\bibitem{hrushovski2018polynomial}
E.~Hrushovski, J.~Ouaknine, A.~Pouly, and J.~Worrell.
\newblock Polynomial invariants for affine programs.
\newblock In {\em Proceedings of the 33rd Annual ACM/IEEE Symposium on Logic in Computer Science}, pages 530--539, 2018.

\bibitem{hrushovski2023strongest}
E.~Hrushovski, J.~Ouaknine, A.~Pouly, and J.~Worrell.
\newblock On strongest algebraic program invariants.
\newblock {\em Journal of the ACM}, 70(5):1--22, 2023.

\bibitem{humenberger2022LoopSynthesis}
A.~Humenberger, D.~Amrollahi, N.~Bj\o{}rner, and L.~Kov\'{a}cs.
\newblock Algebra-based reasoning for loop synthesis.
\newblock {\em Form. Asp. Comput.}, 34(1), July 2022.

\bibitem{karr1976affine}
M.~Karr.
\newblock Affine relationships among variables of a program.
\newblock {\em Acta Informatica}, 6:133--151, 1976.

\bibitem{kempf_1993}
G.~Kempf.
\newblock {\em Algebraic Varieties}.
\newblock London Mathematical Society Lecture Note Series. Cambridge University Press, 1993.

\bibitem{ISSAC2023Laura}
G.~Kenison, L.~Kov\'{a}cs, and A.~Varonka.
\newblock From polynomial invariants to linear loops.
\newblock In {\em Proceedings of the 2023 International Symposium on Symbolic and Algebraic Computation}, ISSAC '23, page 398–406, New York, NY, USA, 2023. Association for Computing Machinery.

\bibitem{kovacs2008reasoning}
L.~Kov{\'a}cs.
\newblock Reasoning algebraically about p-solvable loops.
\newblock In {\em International Conference on Tools and Algorithms for the Construction and Analysis of Systems}, pages 249--264. Springer, 2008.

\bibitem{kovacs2023algebra}
L.~Kov{\'a}cs.
\newblock Algebra-based loop analysis.
\newblock In {\em Proceedings of the 2023 International Symposium on Symbolic and Algebraic Computation}, pages 41--42, 2023.

\bibitem{Lo1983}
R.~Loos.
\newblock Computing rational zeros of integral polynomials by p-adic expansion.
\newblock {\em SIAM Journal on Computing}, 12(2):286--293, 1983.

\bibitem{manna2012temporal}
Z.~Manna and A.~Pnueli.
\newblock {\em Temporal verification of reactive systems: safety}.
\newblock Springer Science \& Business Media, 2012.

\bibitem{rodriguez2004automatic}
E.~Rodr{\'\i}guez-Carbonell and D.~Kapur.
\newblock Automatic generation of polynomial loop invariants: Algebraic foundations.
\newblock In {\em Proceedings of the 2004 international symposium on Symbolic and algebraic computation}, pages 266--273, 2004.

\bibitem{rodriguez2007automatic}
E.~Rodr{\'\i}guez-Carbonell and D.~Kapur.
\newblock Automatic generation of polynomial invariants of bounded degree using abstract interpretation.
\newblock {\em Science of Computer Programming}, 64(1):54--75, 2007.

\bibitem{rodriguez2007generating}
E.~Rodr{\'\i}guez-Carbonell and D.~Kapur.
\newblock Generating all polynomial invariants in simple loops.
\newblock {\em Journal of Symbolic Computation}, 42(4):443--476, 2007.

\bibitem{Rou1999}
F.~Rouillier.
\newblock Solving zero-dimensional systems through the rational univariate representation.
\newblock {\em Applicable Algebra in Engineering, Communication and Computing}, 9(5):433--461, 1999.

\bibitem{Sei1954}
A.~Seidenberg.
\newblock A new decision method for elementary algebra.
\newblock {\em Annals of Mathematics}, 60(2):365--374, 1954.

\bibitem{shafarevich1994basic}
I.~R. Shafarevich and M.~Reid.
\newblock {\em Basic algebraic geometry}, volume~1.
\newblock Springer, 1994.

\bibitem{Shla2011}
A.~Shlapentokh.
\newblock Defining integers.
\newblock {\em The Bulletin of Symbolic Logic}, 17(2):230–251, 2011.

\bibitem{saurabh2010}
S.~Srivastava, S.~Gulwani, and J.~S. Foster.
\newblock From program verification to program synthesis.
\newblock {\em SIGPLAN Not.}, 45(1):313–326, Jan. 2010.

\bibitem{Tar1951}
A.~Tarski and J.~C.~C. McKinsey.
\newblock {\em A Decision Method for Elementary Algebra and Geometry}.
\newblock University of California Press, dgo - digital original, 1 edition, 1951.

\bibitem{von2013modern}
J.~Von Zur~Gathen and J.~Gerhard.
\newblock {\em Modern computer algebra}.
\newblock Cambridge university press, 2013.

\end{thebibliography}

\medskip
{\small\noindent {\bf Authors' addresses}
\medskip

\noindent{Erdenebayar Bayarmagnait,  KU Leuven}\hfill {\tt erdenebayar.bayarmagnai@kuleuven.be}
\\
\noindent{Fatemeh Mohammadi, 
KU Leuven} \hfill {\tt fatemeh.mohammadi@kuleuven.be}
\\ 
\noindent{R\'emi Pr\'ebet,  Inria} \hfill {\tt remi.prebet@ens-lyon.fr}
\\

\end{document}